\documentclass[sigconf]{acmart}

\usepackage{subfigure}
\usepackage{threeparttable}
\usepackage{multirow}
\usepackage{algorithm}
\usepackage{algorithmicx}
\usepackage{algpseudocode}
\usepackage{array}
\usepackage{balance}

\renewcommand\footnotetextcopyrightpermission[1]{}

\newtheorem{observation}{Observation}[section]

\AtBeginDocument{%
	\providecommand\BibTeX{{%
			\normalfont B\kern-0.5em{\scshape i\kern-0.25em b}\kern-0.8em\TeX}}}

\begin{document}

\title{Crowdsourcing-based Multi-Device Communication Cooperation for Mobile High-Quality Video Enhancement}

\fancyhead{}

\author{Xiaotong Wu}
\email{wuxiaotong@njnu.edu.cn}
\affiliation{%
\institution{Nanjing Normal University}
\city{Nanjing}
\country{China}
}

\author{Lianyong Qi}
\email{lianyongqi@gmail.com}
\affiliation{%
	\institution{Qufu Normal University}
	\city{Rizhao}
	\country{China}
}

\author{Xiaolong Xu}
\email{njuxlxu@gmail.com}
\affiliation{%
	\institution{NUIST}
	\city{Nanjing}
	\country{China}
}

\author{Shui Yu}
\email{shui.yu@uts.edu.au}
\affiliation{%
  \institution{University of Technology Sydney}
  \city{Sydney}
  \country{Australia}}

\author{Wanchun Dou}
\email{douwc@nju.edu.cn}
\authornotemark[1]
\affiliation{%
 \institution{Nanjing University}
 \city{Nanjing}
 \country{China}}

\author{Xuyun Zhang}
\email{xuyun.zhang@mq.edu.au}
\affiliation{%
  \institution{Macquarie University}
  \city{Sydney}
  \country{Australia}
}
\authornote{W. Dou and X. Zhang are the corresponding authors of this research.}



\begin{abstract}
  The widespread use of mobile devices propels the development of new-fashioned video applications like 3D (3-Dimensional) stereo video and mobile cloud game via web or App, exerting more pressure on current mobile access network. To address this challenge, we adopt the crowdsourcing paradigm to offer some incentive for guiding the movement of recruited crowdsourcing users and facilitate the optimization of the movement control decision. In this paper, based on a practical 4G (4th-Generation) network throughput measurement study, we formulate the movement control decision as a cost-constrained user recruitment optimization problem. Considering the intractable complexity of this problem, we focus first on a single crowdsourcing user case and propose a pseudo-polynomial time complexity optimal solution. Then, we apply this solution to solve the more general problem of multiple users and propose a graph-partition-based algorithm. Extensive experiments show that our solutions can improve the efficiency of real-time D2D communication for mobile videos.
\end{abstract}



\maketitle

\section{Introduction}
Video streaming has obtained a significant increase of popularity and become the dominant application over the Internet in past decades. According to recent reports \cite{sandvine2016global}, online videos via web or App currently occupy more than half of the Internet traffic. In this process, the widespread use of powerful mobile devices such as smartphones has a far-reaching effect on the prosperity of video streaming. Portable mobile devices enable users to instantly share user-generated videos anywhere over social media applications and enjoy high-quality on-demand videos anytime. The prosperity of video streaming propels the pursuit for better visual experience, boosting the advent of UHD (Ultra-High-Definition) video and new multimedia applications like 3D stereo video \cite{aflaki2015subjective} and cloud game \cite{illahi2020cloud}. These new-fashioned video applications trigger higher bandwidth demand and exert more pressure on current Internet network and mobile access network. The mobile communication sector has been characterized by an exponentially increasing traffic demand for high quality mobile multimedia services \cite{de2016mobile}.


To cater for the development of mobile video streaming, new network technologies such as SDN (Software Defined Networking) \cite{kreutz2014software} and NFV (Network Function Virtualization) \cite{mijumbi2015network} have been studied to optimize Internet traffic and demonstrated to act as service cornerstone on future Internet. Geo-distributed cloud computing platforms \cite{zheng2016online} are introduced to address dynamic video demand scale by leveraging its elastic resource provisioning and seemingly unlimited computation power. However, these technologies focus on the optimization inside Internet network and lack attention to mobile access network. Actually, supporting efficient high-quality real-time video delivery to mobile users in mobile networks is challenging and user perceived quality for video services delivered over these networks is in general low. From Shannon's capacity formula, the maximum transmission rate in wireless communication would be restricted by the limited signal bandwidth and transmit power \cite{zhang2017joint}. The competition among mobile users associated with the same BS (Base Station) limits each individual maximum throughput, lagging behind the users' booming high-quality video traffic demand.

D2D (Device-to-Device) communication paradigm is treated as a promising technology to compensate this gap \cite{li2014message}. D2D communication in cellular networks allows direct communication between two mobile users without traversing a BS or a core network. Most existing works \cite{wu2018cache, wu2019beef, chen2017cache, yan2018network} showed that D2D communication improves spectral and energy efficiency, delay, or even fairness. However, existing D2D communication researches usually build on users' random mobility or position snapshot in cellular networks. This opportunistic mobility characteristic implies that the D2D communication is low-efficient, especially for video transmission. The random mobility of mobile users leads to dynamic variation of communication duration and quality, which can only support video transmission on a best-effort basis. Actually, sustainable real-time D2D communication is preferable for high-quality online videos.

\begin{figure}[t]
	\centering
	\includegraphics[width=8cm,height=4cm]{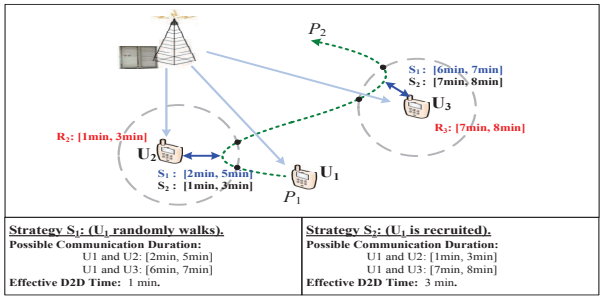}
	\caption{Motivation example: Comparison of video enhancement service quality under two different strategies. 
	}
	\label{fig:moti}
\end{figure}

In this paper, we focus on efficiency improvement of D2D communication for online videos. By investigating a simple example shown in Fig.~\ref{fig:moti}, we observe an important phenomenon that the D2D communication efficiency can be greatly improved by adopting the crowdsourcing paradigm to facilitate the control on users' movement. In detail, there are three smartphone users $U_1, U_2$, and $U_3$. $U_2$ requests the video enhancement service from 1min to 3min ($x$min denotes the time stamp with minute as the unit), while $U_3$ requests the same service from 7min to 8min. Meanwhile, $U_1$ moves from $P_1$ (i.e., a position) to $P_2$ and can support video enhancement service by D2D communication. Assume we examine two video enhancement strategies $S_1$ and $S_2$. In Strategy $S_1$, $U_1$ randomly walks from $P_1$ to $P_2$ and no control is imposed on the movement. The total effective service time of both $U_2$ and $U_3$ is only $1$ minute. In Strategy $S_2$, $U_1$ is recruited and the corresponding total effective service time is 3 minutes. Hence, it can be observed that applying a control on the movement of the recruited mobile users (i.e., $S_2$) can increase the real-time D2D communication efficiency, hence improving the live video quality.

The abovementioned example motivates us to investigate the crucial problem of optimizing the movement of recruited mobile users to improve video quality. Meanwhile, the pioneering work \cite{lin2018sybil, karaliopoulos2019optimal} has demonstrated that adopting incentive mechanisms is an effective measure to guide users to a given destination so as to enhance the video service quality. However, little work focuses on optimizing recruited crowdsourcing users' movement for video enhancement, mainly due to the following challenges. First, it is difficult to implement a quantization analysis of the influence of current real-time D2D communication on high-quality videos on the mobile devices. Second, it is possible that different requests can conflict with each other with respect to the service time, while strong timeliness is essential for a request. Therefore, the far-sighted cooperation decision needs to be made for multiple recruited crowdsourcing users to optimize the efficiency of the system. Third, the number of both requesters and recruited mobile users can be large, which leads to exponential time complexity if one uses a straightforward exhaustive search method. 

In this paper, we first study on practical 4G network throughput measurement. Then, based on the analysis results, we formulate a cost-efficient user recruitment problem (CURP) and design efficient algorithms for both single and multiple recruited crowdsourcing users cases. In summary, the main technical contributions of this paper are summarized as follows:
\begin{itemize}
	\item We analyse the possibility to improve the quality of mobile videos by exploring current mobile network throughput. We dig out important observations: 
	One recruited 4G mobile user can enhance fluent playback of 4K\footnote{4K refers to 3840×2160 pixels, which is defined by the International Telecommunication Union (ITU).} video and a better no-stutter playback of FHD online video.
	\item We formulate the CURP and rigorously prove its intractable computation complexity, i.e., being NP-hard. We design a pseudo-polynomial time optimal algorithm for single crowdsourcing user case and a Graph-Partition-based Algorithm (GPA) for multiple crowdsourcing users.
	\item We perform extensive experiment evaluations on real data sets and the results demonstrate that by leveraging the crowdsourcing technique and facilitating a control on the movement of recruited mobile users, our solution can effectively  guarantee high online video quality. 
\end{itemize}

The rest of this work is organized as follows. Section \ref{sec:related_work} reviews the related work. Section \ref{sec:network_analysis} studies a 4G network throughput analysis. Section \ref{sec:framework} formulates the cost-constrained user recruitment problem. Section \ref{sec:single_user} designs the optimal algorithm for the single user case while Section \ref{sec:general_users} addresses the general problem. Section \ref{sec:perform} evaluates our solutions and Section \ref{sec:conclude} concludes this paper.


\section{Related Work}
\label{sec:related_work}



D2D communication is a promising technology to relieve the pressure forced by the global video streaming on the access network. In academia, D2D communication was first proposed in \cite{lin2000multihop} to enable multi-hop relays in cellular networks. This architecture is demonstrated to be able to improve spectral efficiency of cellular networks dramatically. Then, D2D communication was introduced into video dissemination \cite{le2015microcast, wu2018cache, yan2018network}. 
Based on a network utility maximization framework, Le et al. \cite{le2015microcast} designed a cooperative system called MicroCast, in which each mobile device has simultaneously two network interfaces.
Wu et al. \cite{wu2018cache} proposed a user-centric video transmission mechanism based on D2D communications that allows mobile users to cache and share videos in a cooperative manner.
Yan et al. \cite{yan2018network} proposed a network coding aided collaborative real-time scalable video transmission in D2D communications.

However, existing D2D communications for video dissemination usually assume that some users cache and share these video clips when they meet. 
This best-effort style makes the efficiency of D2D communication low, especially for live video streaming with hard timeliness. 
Fortunately, crowdsourcing paradigm is a potential technique to adopt incentive mechanisms to guide the movement of recruited users to improve the communication efficiency of the whole system \cite{lin2018sybil, karaliopoulos2019optimal}. Lin et al. \cite{lin2018sybil} desinged Sybil-proof online incentive mechanisms to guide users' operations and perform the corresponding tasks. Therefore, we introduce the paradigm to improve the D2D communicaiton efficiency of mobile videos. 


\section{4G Network Throughput Analysis}
\label{sec:network_analysis}

\begin{table}[b]
	\centering
	\begin{threeparttable}[b]
		\caption{Fluent Playback Probability of FHD and 4K Video ("1x" means a single 4G network, while "2x" means the combination of 4G networks via Wi-Fi Direct).}
		\renewcommand\arraystretch{1.15}
		\begin{tabular}{p{0.6cm}<{\centering}p{0.6cm}<{\centering}p{0.75cm}<{\centering}p{0.7cm}<{\centering}p{0.7cm}<{\centering}p{0.7cm}<{\centering}p{0.75cm}<{\centering}p{0.7cm}<{\centering}}
			\toprule
			Video               & Conf.  &  Bicycle & Bus    & Car    & Foot   & Train   & Tram            \\
			\midrule
			\multirow{2}{*}{FHD}&  1x    &  0.9470  & 0.9733 & 0.9746 & 0.9256 & \textbf{0.8853}  & 0.9259            \\
			&  2x    &  0.9983  & 0.9995 & 0.9995 & 0.9974 & \textbf{0.9891}  & 0.9964\\
			\midrule
			\multirow{2}{*}{4K} &  1x    &  0.6409  & 0.7220 & 0.7862 & 0.6182 & \textbf{0.4941}  & 0.6345            \\
			&  2x    &  0.9581  & 0.9824 & 0.9891 & 0.9446 & \textbf{0.8830} & 0.9297\\
			\bottomrule
		\end{tabular}
		\label{table:playback_prob}
	\end{threeparttable}
\end{table}

\subsection{Online Video Watching}


We introduce a real measurement evaluation conducted recently by J. van der Hooft et al. \cite{van2016http} in the city of Ghent, Belgium. 
The dataset contains throughput logs in six different scenes, i.e., bicycle, bus, car, foot, train, and tram.
We focus on two different types of videos, i.e., FHD and 4K video. 
From some popular test movies such as Netflix’s El Fuente \cite{netflix}, a nominal bitrate of 5.2Mbps is enough for mobile user to enjoy a FHD 1080p video by using H.265/HEVC encoding technology while 21.4Mbps are enough to enjoy a 4K 2160p video. Based on these nominal values and the above real 4G throughput, the fluent playback probabilities of these two types of videos are calculated and the corresponding results are listed in Table \ref{table:playback_prob} (Conf. = 1x). An important finding is that one individual’s 4G network can effectively support the fluent playback of FHD video. However, for 4K videos, the situation turns to be very bad. In particular, the fluent playback probability is about two thirds in general and even less than one half in the train scene.



\begin{observation}
	\label{obs:1}
	One user's 4G network effectively supports the fluent playback of online FHD video, but is not good at 4K video.
\end{observation}

\begin{figure}[t]
	\centering
	\includegraphics[width=7cm,height=3.5cm]{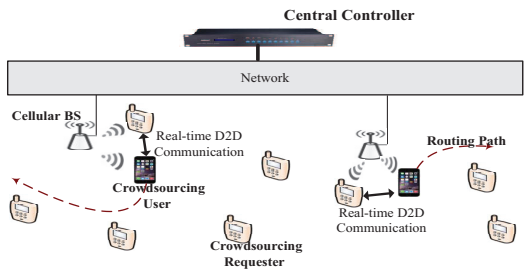}
	\caption{A crowdsourcing scenario with multi-device cooperation.}
	\label{fig:problem_def}
\end{figure}

\subsection{Real-Time D2D Communication}
We consider the influence of real-time D2D communication between mobile users on online videos. OFDM (Orthogonal Frequency Division Multiplexing) technology is usually adopted to proceed data transmission and communication and there would be little interference among mobile users even associated with the same cellular base \cite{cimini1985analysis}. 
\textsc{Proposition} \ref{pro:1} shows one mobile user's fluent playback probability is improved exponentially with the increase of mobile users building real-time D2D communication.

\begin{proposition}
	\label{pro:1}
	Suppose 4G mobile users' network bandwidth conditions are independent identically distributed with CDF of $F$, then the CDF (denote as $G$) of the total available network bandwidths of these $z$ users satisfies the inequality: $ G(x)\le F^{z}(x) $.
\end{proposition}
\begin{proof}
	Without loss of generality, we only need to consider the case of two 4G users. Then,
	$$G(x)=\int_{0}^{x}\int_{0}^{x-t}f(y)f(t)dydt\le\bigg(\int_{0}^{x}f(t)dt\bigg)^2=F^{2}(x).$$
	Then, an iteration process on users' number can complete the proof.
\end{proof}

From \textsc{Proposition} \ref{pro:1}, given any demand $x_{0}$, $1-G(x_{0})\ge 1-F^{z}(x_{0})$. Therefore, for any one mobile user, if all other $z-1$ mobile users build real-time D2D communication with him, his fluent playback probability can be improved dramatically. However, overlarge $z$ is impractical in reality as one smartphone's interface is usually limited (one cellular interface and one WiFi interface). Therefore, we focus on a practical case with $z=2$ where two smartphones are connected through Wi-Fi Direct\footnote{Wi-Fi Direct a recent communication technology superior to Bluetooth. The previous measurement shows Wi-Fi Direct performs 30 times better than Bluetooth in terms of throughput and can allow a high transmission throughput (18$\sim$30 Mbps), depending on communication distance (1$\sim$10 m) \cite{zhanikeev2013virtual}.}. We quantitatively evaluate one mobile user's fluent playback probability in such a cooperative occasion. 
The influence of this cooperation style to online video support is listed in Table \ref{table:playback_prob} (Conf. = 2x). Here, we find that the cooperation of two 4G users can allow one of them to enjoy fluent playback of 4K with a probability of above 90 percent. Moreover, for FHD video, the fluent playback probability improves greatly, especially in the train scene.

\begin{observation}
	\label{obs:2}
	One 4G mobile user can enjoy fluent playback of 4K video with one more 4G mobile user building real-time D2D communication through Wi-Fi Direct. Moreover, he can enjoy a better no-stutter playback of FHD online video by doing this.
\end{observation}

From \textsc{Observation} \ref{obs:1} and \ref{obs:2}, when some smartphone user (2/3/4G) wants to enjoy an excellent view experience of online video, one more 4G user can be employed to achieve this goal by building real-time D2D communication.


\section{Online Video Enhancement via Multi-Device Cooperation}
\label{sec:framework}


\subsection{System Model}

Fig. \ref{fig:problem_def} shows a crowdsourcing scenario where mobile users seek to enjoy enhanced online video view experience by recruiting idle 4G mobile users. We consider a spatial area $\Lambda$ with a set of smartphone users where user channel interference is negligible. The area $\Lambda$ is divided into $K$ discrete non-overlap smaller regions, i.e., $\Lambda=\{1, 2,\dots, K\}$. We assume that crowdsourcing user recruitment is activated at periodic time intervals of duration $T$. Within $T$, suppose $I$ mobile users are \emph{crowdsourcing requesters} while $J$ mobile users are idle and serve as \emph{crowdsourcing users}. Here, one crowdsourcing requester can enjoy an enhanced view experience by connecting to another crowdsourcing user through Wi-Fi Direct. As one requester is relatively still when he/she is watching an online video, we assume crowdsourcing requester $i$'s location is fixed at $d_{i}$. Further, the service range of crowdsourcing user $j$ in region $k$ is denoted as $S_{j,k}$. For simplicity, we assume $S_{j,k}=\{k\}$. That is, the scale of each region is comparable to one crowdsourcing user's service range.

In this model, by leveraging some incentives \cite{lin2018sybil, karaliopoulos2019optimal}, one crowdsourcing user can move to some pre-determined regions. This controlled mobility characteristic results in the expansion of one crowdsourcing user's service range and thus improves the system's performance. The moving cost between different regions is denoted as a matrix \textbf{P} where the item $p_{k_1,k_2}$ stands for moving cost from region $k_1$ to region $k_2$, where $p_{k_1,k_2}=0$ if $k_1=k_2$. It satisfies the following triangle inequality:
\begin{flalign}
p_{k_1,k_2}\le p_{k_1,k_3}+p_{k_3,k_2},\ \forall k_1,k_2,k_3\in\Lambda.
\end{flalign}

Moreover, in this paper, we assume when a crowdsourcing user is in transfer, he would not serve any requester and all crowdsourcing users obey a \emph{homogenous-constant-speed} mobility model. Here, \emph{homogeneity} means different crowdsourcing users need the same time to move from one specific region to another specific region while \emph{constant-speed} means transfer time between different regions also satisfies triangle inequality. Another matrix \textbf{Q} is used to denote transfer time between different regions where the item $q_{k_1,k_2}$ represents transfer time from region $k_1$ to region $k_2$, where $q_{k_1,k_2}=0$ if $k_1=k_2$. The following triangle inequality holds:
\begin{flalign}
q_{k_1,k_2}\le q_{k_1,k_3}+q_{k_3,k_2},\ \forall k_1,k_2,k_3\in\Lambda.
\end{flalign}

Route and association optimization is conducted by a central controller. The controller can collect location information of all crowdsourcing requesters and users, and make a decision on system performance optimization.

\subsection{Problem Formulation}
We formulate the optimization problem that \textbf{the trusted central controller} solves. Suppose the enhanced time interval for requester $i$ starts from $s_i$ and end at $e_{i}$. Let $a_{tj}\in\{0,1\}$ denote whether crowdsourcing user $j$ is in transfer at time $t$ and $y_{tj}^{i}\in\{0, 1\}$ denote whether the request of $i$ is done on crowdsourcing user $j$ at time $t$. Here, we assume that at $t$, one crowdsourcing user serves only one requester through Wi-Fi Direct. $y_{t}^{i}=\sum_{j} y_{tj}^{i} \in \{0, 1\}$ indicates whether the request of $i$ is done at $t$. 
We also give the definitions of requester utility and system utility, respectively.

\begin{definition}[Requester Utility] The requester utility $U_{i}$ of crowdsourcing requester $i$ is defined as the length of service interval i.e., $U_{i}=|\{ t| y_{t}^{i}=1\}|$. For example, if $\{ t| y_{t}^{i}=1\}=[1,2]\cup[3,5]$, $U_{i}=3$.
\end{definition}


\begin{definition}[System Utility] The utility of the whole system is defined as the sum of individual requester utility i.e., $U=\sum_{i} U_{i}$.
\end{definition}

We use $l^{j}$ and $l_{t}^{j}$ to stand for routing path (region sequence) of crowdsourcing user $j$ and its location (region) at $t$ respectively. Specially, $l_{0}^{j}$ means the initial location. The symbol $\mathcal{C}(j\xrightarrow{t} l^{j})$ counts the relevant moving cost when $j$ follows the path $l^{j}$ until time $t$. That is,
\begin{equation}
\mathcal{C}(j\xrightarrow{t} l^{j})=\left\{
\begin{aligned}
&0,\ &&t=0\\
&\mathcal{C}\big(j\xrightarrow{b(t)} l^{j}\big)+p_{l_{b(t)}^{j},l_{t}^{j}},\ &&t>0\\
\end{aligned}\right.
\end{equation}
where $b(t)=t^{-}$ means the previous moment before $t$.

Here, we formulate crowdsourcing user recruitment problem aiming at maximizing the system utility $U$. Actually, there is a relationship between system utility and the maximum available moving budget. A higher latter can support a higher system utility. Thus, we study the cost-constrained situation where central controller solves the following problem:
\begin{definition}[Cost-constrained User Recruitment Problem, CURP] The problem aims at maximizing system utility within a limited budget by making routing path and association decisions, i.e.,
	\begin{flalign}
	\max_{\textbf{\emph{l}},\ \textbf{\emph{y}}}\ \ \ \ &\ \ \ \ \ \ \ \ \ \ U\\
	\textrm{\emph{s.\ t.}}\ \ \ \ \  &\sum_{j}\ \mathcal{C}\big(j\xrightarrow{T} l^{j}\big)\le C\\
	&\sum_{i}y_{tj}^{i}\le 1,\ \ \forall t,j\\
	&\ y_{tj}^{i}\le 1-a_{tj},\ \forall t,j,i\\
	&\ y_{tj}^{i}=0,\ \ \forall d_{i}\notin S_{j,l_{t}^{j}}\\
	&\ y_{tj}^{i}, a_{tj}\in\{0, 1\},\ \forall t,j,i
	\end{flalign}
	where $\textbf{\emph{l}}$$=$$\{l^{j}|\ \forall j\}$, $\textbf{\emph{y}}$$=$$\{y_{tj}^{i}|\ \forall t,j,i\}$ and $C$ is the maximum available moving budget. Eq. (5) ensures a limited cost budget while Eq. (6) means one crowdsourcing user serves at most one requester every moment. Eq. (7) guarantees that crowdsourcing user would not provide service in transfer. Eq. (8) demands one crowdsourcing user should build real-time D2D communication with another requester when they are close enough. Moreover, $C$ is a constant knob controlling the optimal objective value.
\end{definition}


\subsection{Computation Complexity}

\begin{theorem}
	The CURP is NP-hard, i.e., there exists no polynomial
	time optimal algorithm to this problem unless P = NP.
\end{theorem}
\begin{proof}
	We prove that the famous knapsack Problem can be reduced in polynomial time to Cost-Constrained User Recruitment Problem.
	Given one instance of knapsack Problem: There are a set of $M$ items numbered from 1 up to $M$, each with a weight $w_{m}$ and a value $v_{m}$, along with a maximum weight capacity $W$,
	\begin{equation}
	\max \sum_{m}v_{m}x_{m},\ \ s.t.,\ \sum_{m}w_{m}x_{m}\le W; x_{m}\in\{0, 1\}.
	\end{equation}
	For Cost-Constrained User Recruitment Problem, we consider the simplest case with $M$ crowdsourcing requesters and only $1$ crowdsourcing user, i.e., $I=M$ and $J=1$. These crowdsourcing requesters are located in different regions, i.e., $d_{i_1}$$\neq$$d_{i_2}, \forall i_1\neq i_2$. And the enhanced interval satisfies $e_i+q_{max}\le s_{i+1}, \forall i=1,\dots, M-1$ where $q_{max}$ is the maximum value in \textbf{Q}. We make the following settings:
	\begin{itemize}
		\item \textbf{P}: $p_{k,d_{i}}=w_{i}, \forall k=1,2,\dots, K; i=1,2,\dots, M$;
		\item $e_i-s_i = v_{i}, \forall i=1,\dots, M$;
		\item $K=I+1$ and $l_{0}^{1}\neq d_{i}, \forall i=1,\dots, M$;
		\item $C=W$.
	\end{itemize}
	Routing path optimization is simplified to be a node selection decision. Note that this setting would not violate moving cost triangle inequality. To prove this, we randomly pick out three points $k_{1},k_{2},k_{3}$, then $p_{k_1,k_2}+p_{k_2,k_3}=w_{k_{2}}+w_{k_{3}}\ge w_{k_{3}}=p_{k_1,k_3}$.
	
	We use $x_i\in \{0, 1\}$ to denote whether the requester $i$ is served or not. With the above settings, this special case of Cost-Constrained User Recruitment Problem is simplified to be just the instance of Knapsack Problem. Therefore, Cost-Constrained User Recruitment Problem is at least as hard as Knapsack Problem, the optimization version of which is NP-hard.
\end{proof}


\section{Optimal User Recruitment Strategy for Single Crowdsourcing User}
\label{sec:single_user}

One major challenge to solve the CURP lies in infinite decision variables $\textbf{\emph{l}}$, $\textbf{\emph{y}}$ for all $t\in[0, T]$. In this section, we consider a special case of CURP in which $J=1$, i.e., there is only one crowdsourcing user. We first introduce a proposition which captures a useful property of the problem. The property specifies the characteristic of an optimal decision which can reduce infinite decision variables.

\begin{figure}[t]
	\centering
	\includegraphics[width=6.2cm,height=2.7cm]{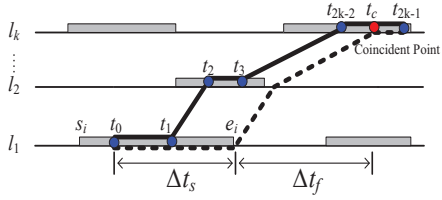}
	\caption{The illustration of proof of \textsc{Proposition} \ref{prop:2}: We compare the system performance between any optimal strategy $\Gamma$ (solid line) and adjusted policy (dotted line).}
	\label{fig:illus_p2} 
\end{figure}

\begin{proposition}
	\label{prop:2}
	Given any problem instance with $J$$=$$1$, there always exists one optimal decision $\Gamma^{*}$, satisfying the following condition: for each requester $i$, $\exists \tau_{i}\in [s_i, e_i]$ such that $y_{t}^{i}=0, \forall t\in [s_i, \tau_i]$ and $y_{t}^{i}=1, \forall t\in [\tau_i, e_i]$.
\end{proposition}
\begin{proof}
	To prove this proposition, we only need to show that any \emph{optimal} decision $\Gamma$ can be adjusted to be $\Gamma^{*}$ without incurring any additional cost. We consider two cases in $\Gamma$.

	\textbf{Case 1.} The crowdsourcing user does not serve $i$: We only treat $\tau_{i}$ as the ending time of $i$'s request, i.e., $\tau_{i}=e_i$.
	
	\textbf{Case 2.} The crowdsourcing user starts serving $i$ at $t_{0}$ where $t_{0}\in [s_{i}, e_i]$: Suppose that in $\Gamma$, the following regions and time intervals the crowdsourcing user works in are $l_1:[t_0,t_1]$, $l_2:[t_2,t_3]$, $\cdots$. We compare $\Gamma$ with another adjusted policy.
	
	Adjusted Policy $\Gamma^{'}$: The crowdsourcing user keeps serving $i$ in $[t_{0}, e_{i}]$ and then transfers to regions as the same order with $\Gamma$ and does not serve any requester until the temporal-spatial coincident point with $\Gamma$. Next, the user works under $\Gamma$.
	
	The comparison is illustrated in Fig. \ref{fig:illus_p2}. To compare $\Gamma$ with this adjusting policy, we only need to consider time interval $[t_0,t_c]$ where $t_c$ is the temporal-spatial coincident point. In $[t_0,t_c]$: As the user transfers to regions in the same order under $\Gamma$ and $\Gamma^{'}$, the transfer cost under $\Gamma^{'}$ is equal to $\Gamma$ and the transfer time $\Delta t_f$ is also equal. In terms of system performance, the utility under $\Gamma^{'}$ is $\Delta t_{s}=e_{i}-t_0=t_c-t_0-\Delta t_f$. The utility under $\Gamma$ would be no more than $\Delta t_s$ depending on whether the user keeps working when not in transfer under $\Gamma$. That is, $\Gamma^{'}$ does not decrease the system utility with incurring no additional cost.
	Finally, this adjusted policy can be adopted in each served request $i$ in $\Gamma$, therefore, by doing this, $\Gamma$ is adjusted to $\Gamma^{*}$-type strategy with $\tau_i=t_{0}$.
\end{proof}

From \textsc{Proposition} \ref{prop:2}, when the crowdsourcing user moves to a region to serve one request, it will not leave unless this request has been finished. That is, the crowdsourcing user's location only changes at a limited number of time points. This result will help us design an optimal dynamic programming algorithm. We introduce one virtual request $v_{0}$ with $s_{v_{0}}=e_{v_{0}}=0$. This virtual request stands for the crowdsourcing user's starting location. Then all requests (including $v_{0}$) are sorted according the ascending order of the ending time of these requests. Assume that the set of sorted tasks is $\mathcal{V}=\{v_{0}, v_{1},\dots,v_{I}\}$ and the ending times satisfy $e_{0}\le e_{1}\le \dots\le e_{I}$.

Let $\Phi_{i}$ be the set containing all requests whose ending times plus transfer time are smaller than $v_{i}$'s ending time. Precisely, we have $\Phi_{i}=\{v_{j}|e_{j}+q_{ji}\le e_{i}, j=0, 1,\dots, I\}$ and $\Phi_{0}=\{v_{0}\}$. According to \textsc{Proposition} \ref{prop:2}, the crowdsourcing user keeps serving one requester until the request is completed. Therefore, $\Phi_{i}$ actually contains all possible prior one requests before coming to $v_i$.

Based on the above analysis, let $U[i][c]$ be the optimal system utility gain when the crowdsourcing user is at $v_{i}$ with the total cost no more than $c$. We have,
\begin{equation}
U[i][c]=\left\{
\begin{aligned}
&-\infty,&&c<0\ \textrm{or}\ v_0\notin\Phi_{i}\\
&0,&&i=0,c\ge 0\\
&\max_{j\in\Phi_{i}\backslash\{i\}}U[j][c_{ji}]+u_{ji},&&i\ge 1,c\ge 0 \\
\end{aligned}\right.
\end{equation}
where $c_{ji}=c-p_{ji}$, $u_{ji}=e_{i}-\max\{s_i,e_{j}+q_{ji}\}$ and $p_{ji}$ stands for the moving cost from $v_j$'s location to $v_{i}$'s location. Eq. (11) shows how the value of $U[i][c]$ can be recursively computed. When $c<0$, $U[i][c]$ is not feasible because the cost can not be negative. $\Phi_{i}\not\supset\{v_0\}$ means the user can not reach $i$ in time, hence $U[i][c]$ is also infeasible. Considering that the CURP is a maximization problem, we make the infeasible $U[i][c]=-\infty$. Otherwise, when $c\ge 0$ and $\Phi_{i}\supset\{v_0\}$, $U[i][c]$ can be derived from all possible prior one subproblems corresponding to $U[j][c-p_{ji}], j\in\Phi_{i}$. Moving cost inequality Eq. (1) and transfer time inequality Eq. (2) ensure the correctness of Eq. (11). With Eq. (1) and Eq. (2), the crowdsourcing user would transfer from one region to another region directly without searching cost-shortest or time-shortest path.

\begin{algorithm}[tbp]
	\caption{Optimal Routing Scheduling}
	\label{alg:ors} 
	\begin{algorithmic}[1]
		\Require
		\textbf{P}, \textbf{Q}, all sorted requests $v_i$ where $i=0,1,\cdots,I$;
		\Ensure
		Decision variables \textbf{\emph{l}} and \textbf{\emph{y}}.
		\State Calculate prior one request set $\Phi_{i}$ for each $v_{i}$: $\Phi_{0}\leftarrow\{v_0\}$ and $\Phi_{i}\leftarrow\{v_{j}|e_{j}+q_{ji}\le e_{i}, j=0, 1,\dots, I\},\forall i=1,\cdots,I$;
		\State Initialize the boundary condition: $U[i][c]\leftarrow -\infty, \forall c<0$ or $\forall \Phi_{i}\not\supset\{v_0\}$ and $U[0][c]\leftarrow 0, \forall c\ge 0$;
		\State Calculate all other subproblems $U[i][c]$ based on Eq. (11):
		$U[i][c]\leftarrow\max_{j\in \Phi_{i}\backslash\{i\}} U[j][c_{ji}]+u_{ji},\forall i\ge 1,c\ge 0,v_0\in\Phi_{i}$;
		\State \emph{\textbf{l}} and \emph{\textbf{y}} corresponding to $\max_{i\in \{0,1,\cdots,I\}}U[i][C]$;
	\end{algorithmic}
\end{algorithm} 

Based on Eq. (11), we design Algorithm \ref{alg:ors}, an optimal scheduling algorithm. Step 1 computes prior one request sets for all sorted requests while Step 2 and 3 show the complete dynamic programming process. The optimality of this algorithm is analyzed as follows.
\begin{theorem}
	\label{thm:alg_analysis}
	Algorithm \ref{alg:ors} has an optimal solution to the CURP problem ($J$=$1$) under the pseudo-polynomial time computation complexity $\mathcal{O}(I^2(C+p_{max}))$ where $p_{max}$ is the maximum value in \textbf{\emph{P}}.
\end{theorem}
\begin{proof}
The CURP($J$=$1$) problem has optimal substructure. From the definition of $U[i][c]$ and \textsc{Proposition} \ref{prop:2}, there exist $|\Phi_{i}|$ possible prior one states for $U[i][c]$. And the optimality of $U[i][c]$ demands the maximum utility gain of these $|\Phi_{i}|$ subproblems. Therefore, the optimal solution of $U[i][c]$ contains optimal solutions to all these prior one $|\Phi_{i}|$ subproblems. Algorithm 1 explores every subproblem of $U[i][c]$ and computes $U[i][c]$ based on the fact of optimal substructure, hence, it produces an optimal solution to the CURP($J$=$1$) problem.
Eq. (11) indicates $U[i][c](i\le I, c\le C)$ has overlapping subproblems and there exists at most $(I+1)$$\times$$(C+p_{max}+1)$ subproblems according different values of $i$ and $c$. In order to solve each subproblem, we need to compare $|\Phi_{i}|$$(\le I)$ subproblems. Therefore, Algorithm 1 has a pseudo-polynomial time complexity $\mathcal{O}(I^2(C+p_{max}))$.
\end{proof}

\section{General Optimization Strategy}
\label{sec:general_users}

We discuss the general CURP with multiple crowdsourcing users, i.e., $I$$\gg$$J$$\ge$$1$. We design a Graph-Partition-based Algorithm (GPA) by extending the above Optimal Routing Algorithm.
The basic idea of GPA is to \emph{decrease the crowdsourcing user's transfer number on the high-cost link}. Since the total available cost $C$ is limited, decreasing one high-cost link transfer means increasing more low-cost link transfers, hence higher chances of serving more requests. To realize this idea, we adopt a graph-partition method. All requests are divided into several disjoint partitions and one crowdsourcing user is only responsible for one partition. Thus, according to reasonable partition rules, those high-cost links between different partitions can be decreased or even removed. There are two steps of GPA: Requests Partition and Available Cost Allocation. We show the detailed explanations in the following subsections.

\subsection{Request Partition by Normalized Cut Spectral Clustering}
We adopt normalized cut spectral clustering algorithm to do request partition. In the clustering process, it is important to calculate the similarity relation between any two requests. In our problem, we consider that any one good partition should possess the following two properties. First, the requests in one partition should have fewer overlapping parts. Second, one partition has smaller connectivity cost inside. We define two metrics to measure these two properties.

\begin{figure}[t]
	\centering
	\includegraphics[width=6.8cm,height=2.5cm]{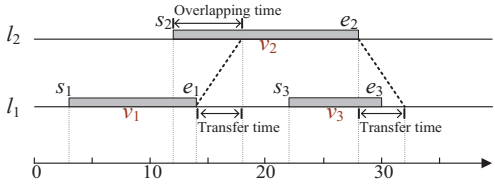}
	\caption{An example to show the calculation and meaning of $\theta$ and $\gamma$.}
	\label{fig:gpa}
\end{figure}

\subsubsection{Overlapping degree} For any two requests $v_{i}$ and $v_j$ with $e_i\le e_j$, we define their \emph{overlapping degree} $\theta$ as follows:

\begin{equation}
\theta(v_i, v_j)\triangleq 1-\frac{(e_j-\max\{s_j,e_i+q_{ij}\})^{+}}{e_j-s_j},
\end{equation}
in which $x^{+}$$=$$\max\{x,0\}$.

\subsubsection{Connectivity degree} The relevant \emph{connectivity degree} $\gamma$ between $v_{i}$ and $v_j$ with $e_i\le e_j$ is defined as  $ \gamma(v_{i},v_j)\triangleq \frac{p_{ij}}{p_{max}} $.

We give an instance to illustrate the calculation of $\theta$ and $\gamma$ defined above. As shown in Fig. \ref{fig:gpa}, there are three requests $v_1, v_2$ and $v_3$. Their request times are $[3, 14]$, $[12, 28]$ and $[22, 30]$. The transfer time between region $l_1$ and $l_2$ is 4 while the moving cost $p_{l_1,l_2}$$=$$p_{l_2,l_1}$$=$$1$. For the pair $(v_1, v_2)$, any crowdsourcing user can only transfer from $v_1$ to $v_2$ based on \textsc{Proposition} \ref{prop:2}. Their maximal non-overlapping time is $e_2$$-$$e_1$$-p_{12}$$=$$10$, and $\theta(v_1,v_2)$$=$$1-10/(28-12)$$=$$0.375$. Likewise, $\theta(v_1,v_3)$$=$$0$ and $\theta(v_2,v_3)$$=$$1$. Since $p_{max}$$=$$1$, therefore, $\gamma(v_1,v_2)$$=$$1$, $\gamma(v_1,v_3)$$=$$0$ and $\gamma(v_2,v_3)$$=$$1$. Obviously, both $\theta$ and $\gamma$ belong to $[0,1]$ and can reflect its meaning truly.

\subsubsection{Similarity weight} The \emph{similarity weight} between $v_{i}$ and $v_{j}$ is calculated by $\alpha_{ij}=\frac{1}{\theta+\beta\gamma}$ where $\beta$ is a weight knob of overlapping degree and connectivity degree.

Let $\textbf{W}=[\alpha_{ij}]$ be the similarity matrix among requests and $\textbf{L}$ be the diagonal matrix whose $i$-th diagonal element is the sum of the elements in the $i$-th row of $\textbf{W}$.
Let $D_{r}$ be the $r$-th partitioned request set and $D=\cup_{r=1}^{J}D_r$. Let $\textbf{e}_{r}$ be a $I\times 1$ indicator vector for the $r$-th partition, i.e., $\textbf{e}_{r}\in \{0,1\}^{I}$, and have a nonzero component exactly when the requester is in the $r$-th partition. Our goal is to find a good partition $\textbf{E}=(\textbf{e}_1,\textbf{e}_2,\dots,\textbf{e}_{J})$, which follows the above two properties, i.e.,
\begin{equation}
\min_{\textbf{E}} \sum_{r=1}^{J}\frac{\textbf{e}_{r}^{T}(\textbf{L}-\textbf{W})\textbf{e}_{r}}{\textbf{e}_{r}^{T}\textbf{L}\textbf{e}_{r}}=\sum_{r=1}^{J}\frac{\sum_{v_{i}\in D_{r}, v_{j}\notin D_{r}}\alpha_{ij}}{\sum_{v_{i}\in D_{r}, v_{j}\in D}\alpha_{ij}}.
\label{eq:good_part}
\end{equation}

The popular normalized spectral clustering algorithm can be adopted to solve this problem. Normally, the normalized spectral clustering algorithm first relaxes the above problem to be an eigenvalue problem and then rounds the relaxation solution to the integer solution satisfying the binary constraint.

\subsubsection{Finding a relaxation optimization solution}
According to \cite{bach2004learning}, one lower bound of Eq. (\ref{eq:good_part}) is given as
\begin{equation}
\min_{\textbf{Y}^{T}\textbf{Y}=\textbf{I}} \texttt{tr} \textbf{Y}^{T}(\textbf{I}-\textbf{L}^{-1/2}\textbf{W}\textbf{L}^{-1/2})\textbf{Y}=\sum_{j=1}^{J}\lambda_{j},
\end{equation}
in which $\textbf{Y}\in R^{I\times J}$ and $\lambda_{j}$ refers to the $j$-th smallest eigenvalue of matrix $\textbf{I}-\textbf{L}^{-1/2}\textbf{W}\textbf{L}^{-1/2}$. This bound can be attained at $\textbf{Y}=\textbf{U}$ where $\textbf{U}$ is a matrix whose columns are the eigenvectors corresponding to the first $J$ eigenvalues of $\textbf{I}-\textbf{L}^{-1/2}\textbf{W}\textbf{L}^{-1/2}$. The relaxation solution of Eq. (\ref{eq:good_part}) can be calculated by $\hat{\textbf{E}}=\textbf{L}^{-1/2}\textbf{U}$.

\subsubsection{Rounding the relaxation solution to binary solution}
$\hat{\textbf{E}}$ provides an approximation of $\textbf{E}$ and there may exist non-integer components in $\hat{\textbf{E}}$. Therefore, how to round $\hat{\textbf{E}}$'s components to be binary should be studied. A simple $k$-means algorithm is adopted to achieve this goal. For each row of $\hat{\textbf{E}}$, it first normalizes the row to norm 1, i.e., $e_{ij}\leftarrow e_{ij}/(\sum_{m=1}^{J}e_{im})^{1/2}$. Then these rows are clustered into $J$ clusters by using a standard $k$-means algorithm.

\subsection{Branch-and-Bound/Greedy Cost Allocation}
Suppose that after normalized cut spectral clustering algorithm, all requests are classified into $J$ partitions, denoted as $\{D_{r}^{*}|\ r=1,..., J\}$. For each crowdsourcing user, the average Euclidean distance to all requests in any partition is calculated. Then the user is dispatched to the partition with the smallest distance. For each generated partition, Algorithm 1 is leveraged to seek the optimal sub-decision. We use $U_{r}[c]$ ($0\le c\le C$) to represent the maximum utility given the available cost $c$ on the $r$-th partition $D_{r}^{*}$.

As the maximum total available cost is limited, we should allocate these available costs to different partitions to maximize the system utility. Let $c_r$ denote the cost allocated on $r$-th partition and this problem can be formulated as follows,
\begin{equation}
\begin{aligned}
\max \sum_{r}U_{r}[c_{r}] \qquad \textrm{s. t.}\; \sum_r c_r\le C.
\end{aligned}
\end{equation}
The problem is a combinatorial optimization problem. We only know that $U_{r}[c]$ is monotonously increasing in $c$ as Algorithm \ref{alg:ors} gives the optimal solution given $c$. Based on this monotonic property, we first propose an optimal branch-and-bound search approach to allocate total available cost.


\subsubsection{Branch-and-bound cost allocation} Without loss of generality, in the search process, we assume that the cost $\{c_1, c_2,\cdots, c_J\}$ on different partitions is sequentially determined. Then, we have the following proposition:

\begin{proposition}
	\label{prop:3}
	$g_k$ is the total amount of system utility gain for the $k$-th partition having been branched and $h_{k}$ is the corresponding cost having been allocated. The upper bound on the total amount of system utility gain for the remaining $J-k$ partitions is $ 	u_{k}=g_k+\sum_{r=k+1}^{J} U_{r}[C-h_{k}] $.
\end{proposition}
\begin{proof}
	$U_{r}[c]$ is monotonously increasing in $c$ and the remaining $J-k$ partitions all take the maximum available cost. Therefore, $u_{k}$ provides an upper bound on the system utility gain for the remaining $J-k$ partitions.
\end{proof}

Based on \textsc{Proposition} \ref{prop:3}, the branching process is a depth-first search over subproblems allocating subset of the available cost, with the ones of the biggest upper bounds being searched first. Each time when a leaf node in the search tree is reached, all other branches with smaller values of upper bounds are pruned. The optimal solution will be found until all the leaf nodes have been either pruned or searched.

\subsubsection{Greedy cost allocation}
Although the branch-and-bound search process efficiently avoids local optimum and achieves the optimal cost allocation, its time consumption seems unacceptable as $C$ may be too large.
To cope with large $C$ case, a greedy approach is designed. Specifically, we first give a maximum step length $c_{0}$ and keep a temporary cost allocation vector $\textbf{c}^{'}=\{c_1^{'},c_2^{'},\cdots,c_J^{'}\}$. Then, we calculate
cost-increment-efficiency of $U_r$ in $[c_i^{'},c_i^{'}+c_0]$ and search the cost increment $(r^{*},c_s^{*})=\arg \max \{\frac{U_r[c_i^{'}+c_s]-U_r[c_i^{'}]}{c_s}| \forall r, c_s\le c_0\}$. Update $c_{r^{*}}^{'}=c_{r^{*}}^{'}+c_s^{*}$. This process is repeated until the cost is beyond the total available cost.


From our experimental evaluation, $U_r[c_r]$ can be treated as an approximate increasing concave function with some small fluctuation.
This characteristic ensures the efficiency of greedy method. Actually, if $U_r[c_r]$ is strict increasing concave, the optimal cost allocation would be achieved at the equal derivative point by solving the KKT (Karush–Kuhn–Tucker) condition, i.e., $\frac{\partial U_m}{\partial c_m}=\frac{\partial U_n}{\partial c_n}$. The greedy method would achieve this optimal solution in this case. However, small fluctuation may break this optimality. In our greedy method, the variable step length in $[1,c_0]$ can relieve the adverse effect of small fluctuation. Therefore, greedy approach can generally achieve a good solution.

\begin{figure}[t]
	\centering
	\subfigure[Random waypoint model]{
		\includegraphics[width=0.225\textwidth]{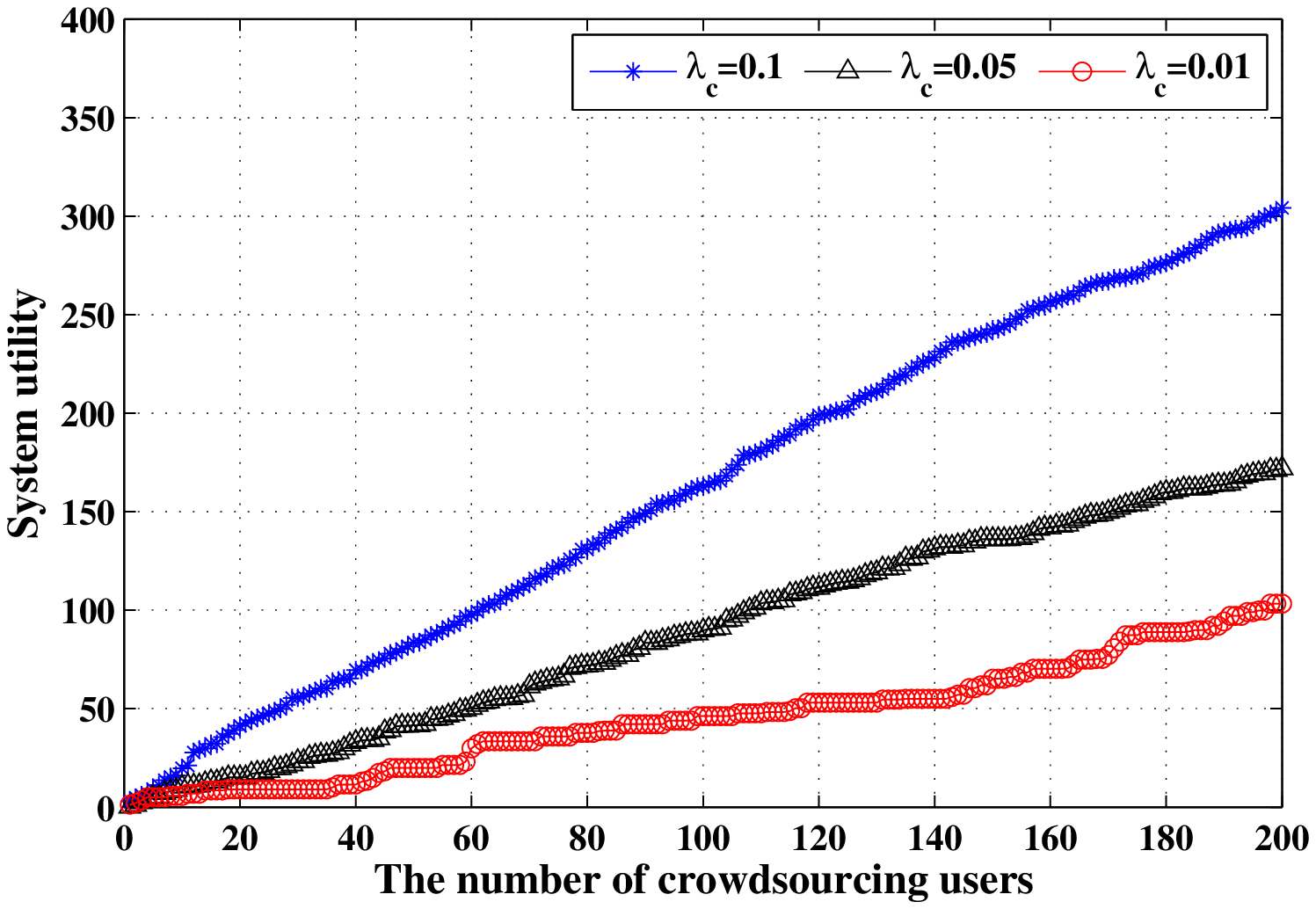}
		\label{fig:rwm}
	}
	\subfigure[Campus waypoint model]{
		\includegraphics[width=0.225\textwidth]{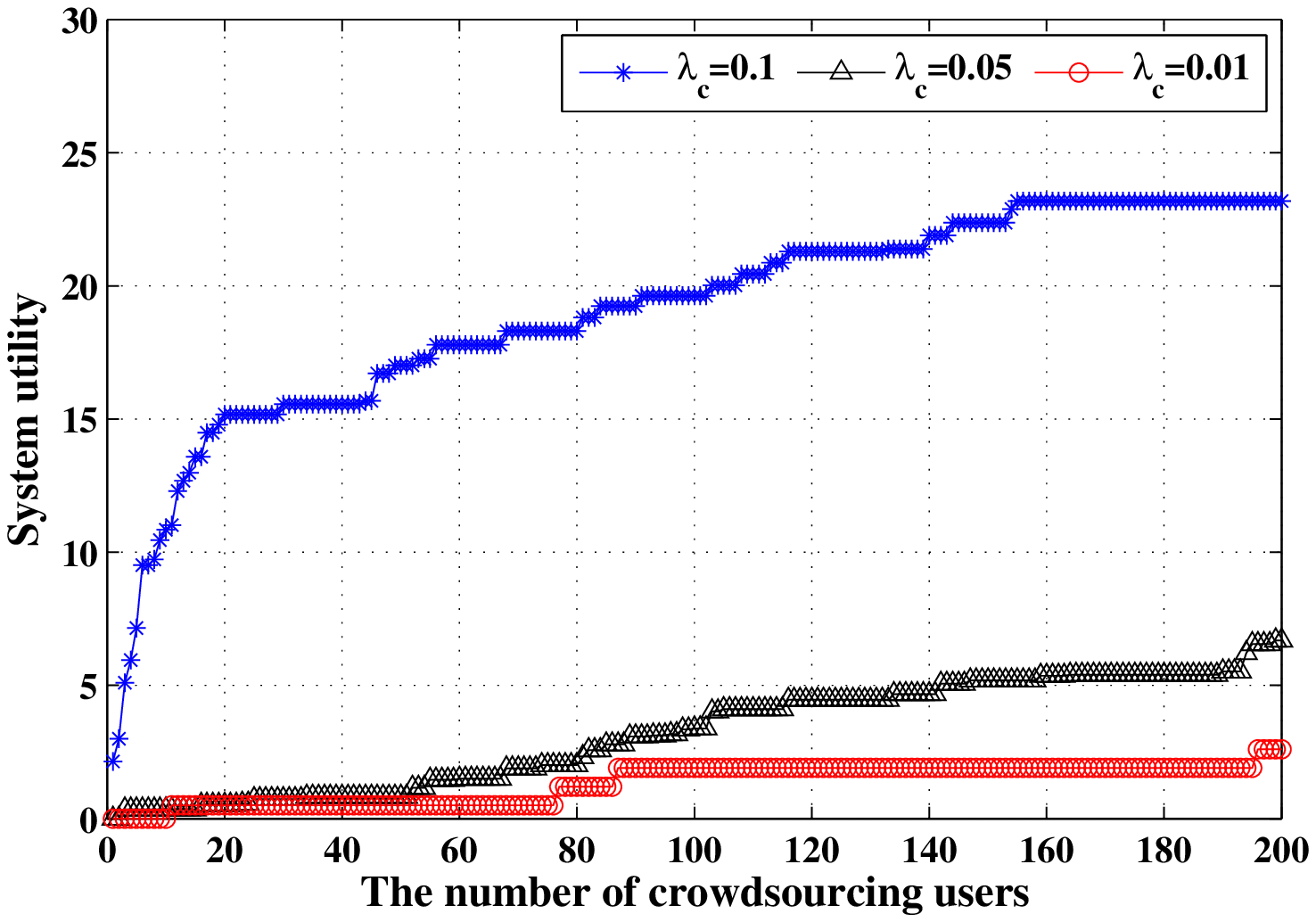}
		\label{fig:cwm}
	}
	\caption{The influence of user scale on system utility under two different mobility models without movement control.}
	\label{fig:wolf2}
\end{figure}

\section{Performance Evaluation}
\label{sec:perform}


\subsection{Simulation Setup}
In our simulation, we consider an area which is divided into $50\times45$ regions and each region is a square with side length of 5 meters. Two mobility models are adopted to simulate the track of users, i.e., \emph{random waypoint model} \cite{broch1998performance} and \emph{campus waypoint model} \cite{mcnett2005access}. These two models depict different occasions: the former one means relatively uniform node densities with simulation areas while the later one represents non-uniform node densities. The model parameters are set based on some well-known benchmarks. Specifically, in these two models, one crowdsourcing user's residence time in one region follows the Pareto distribution \cite{lu2016cooperative} with shape parameter 5 and scale parameter 2 (2 is the minimum possible of the residence time). The average view time per request follows the gamma distribution \cite{zhou2015video} with shape parameter 4 and scale parameter 2. In each region, the requests come following the Poisson process and the total arrival rate $\lambda_{c}$ in this area is set to be in the interval [0.01, 0.1].



\begin{figure}[t]
	\centering
	\subfigure[Available cost]{
		\includegraphics[width=0.225\textwidth, height=0.12\textheight]{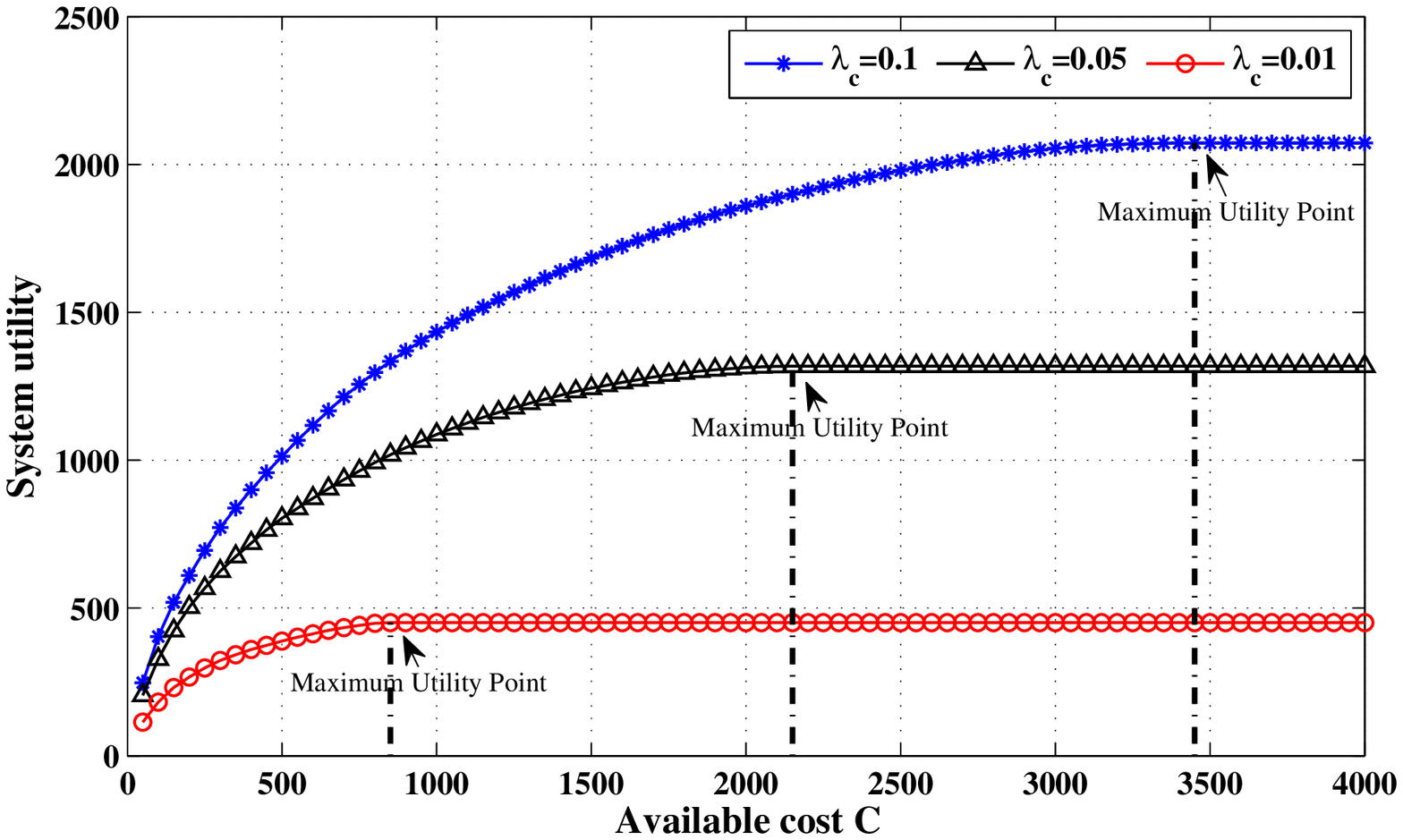}
		\label{fig:cost_change}
	}
	\subfigure[Total available cost]{
		\includegraphics[width=0.225\textwidth, height=0.12\textheight]{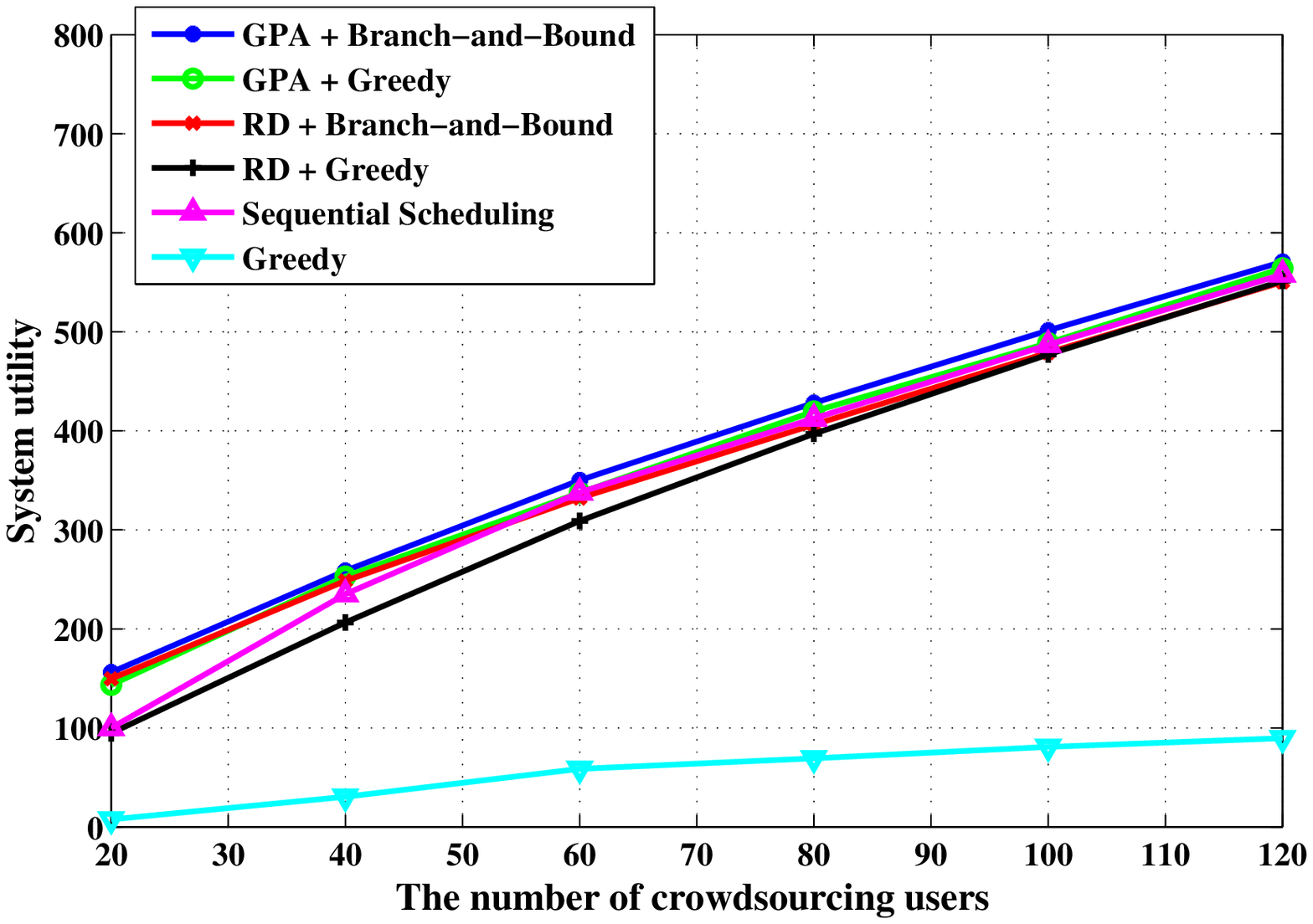}
		\label{fig:utility_cost}
	}
	\caption{The influence of available cost on system utility. The left-figure shows random waypoint model with movement control. The right-figure evaluates the total available cost $C$ on the optimality with $\lambda_c=0.1, J=3$.}
	\label{fig:cost_effect}
\end{figure}

In order to do comparative analysis, we realize some other benchmark algorithms besides our algorithm. The corresponding algorithmic are \emph{Region Division (RD) + Branch-and-Bound/Greedy Cost Allocation}, \emph{Sequential Scheduling} and \emph{Greedy Scheduling}.

\begin{figure*}[t]
	\centering
	\subfigure[System utility with the increase of $J$]{
		\label{fig:utility_n} 
		\includegraphics[width=5.7cm,height=0.11\textheight]{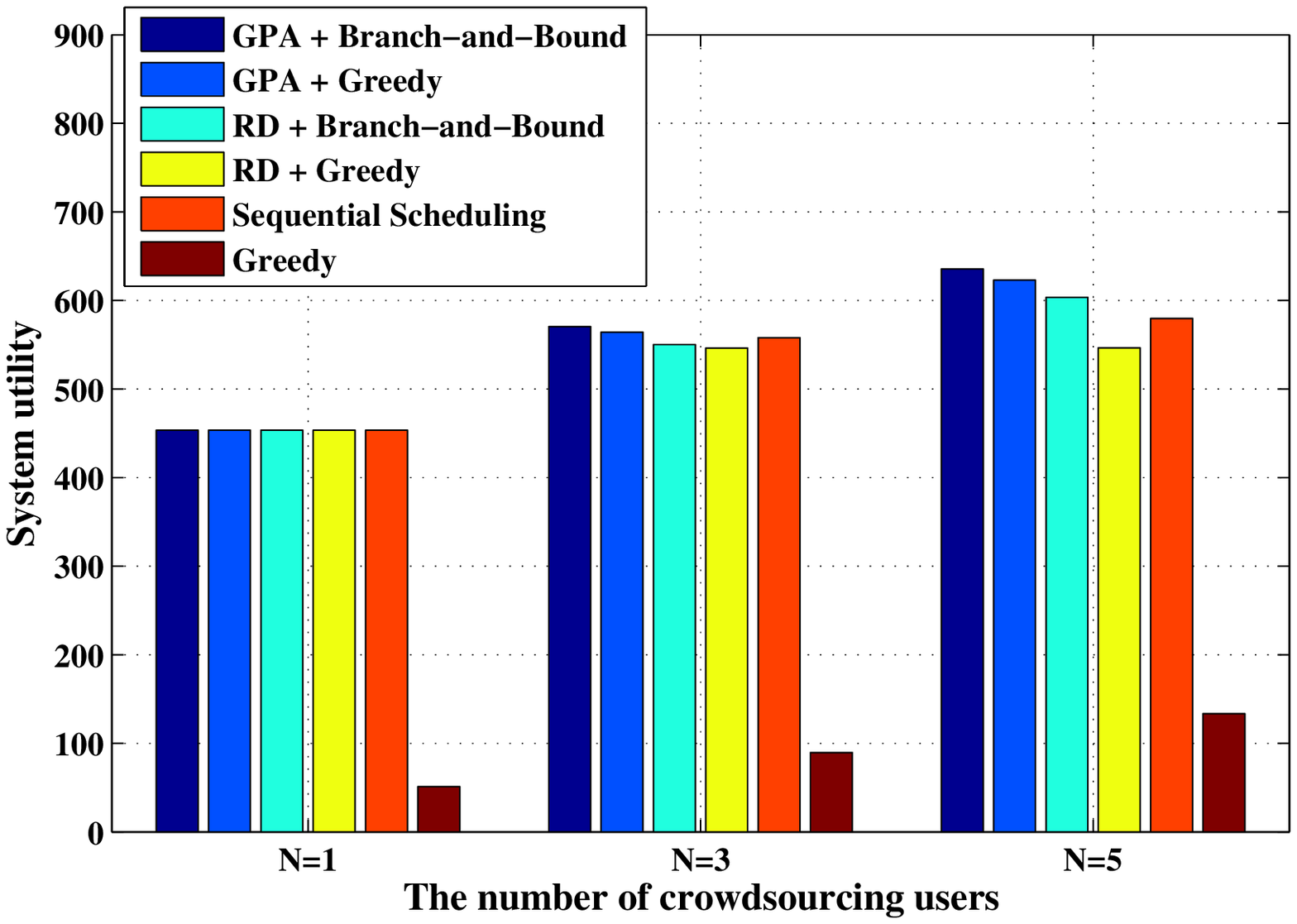}}
	\hspace{1pt}
	\subfigure[Time consumption with the increase of $J$]{
		\label{fig:time_n} 
		\includegraphics[width=5.7cm,height=0.11\textheight]{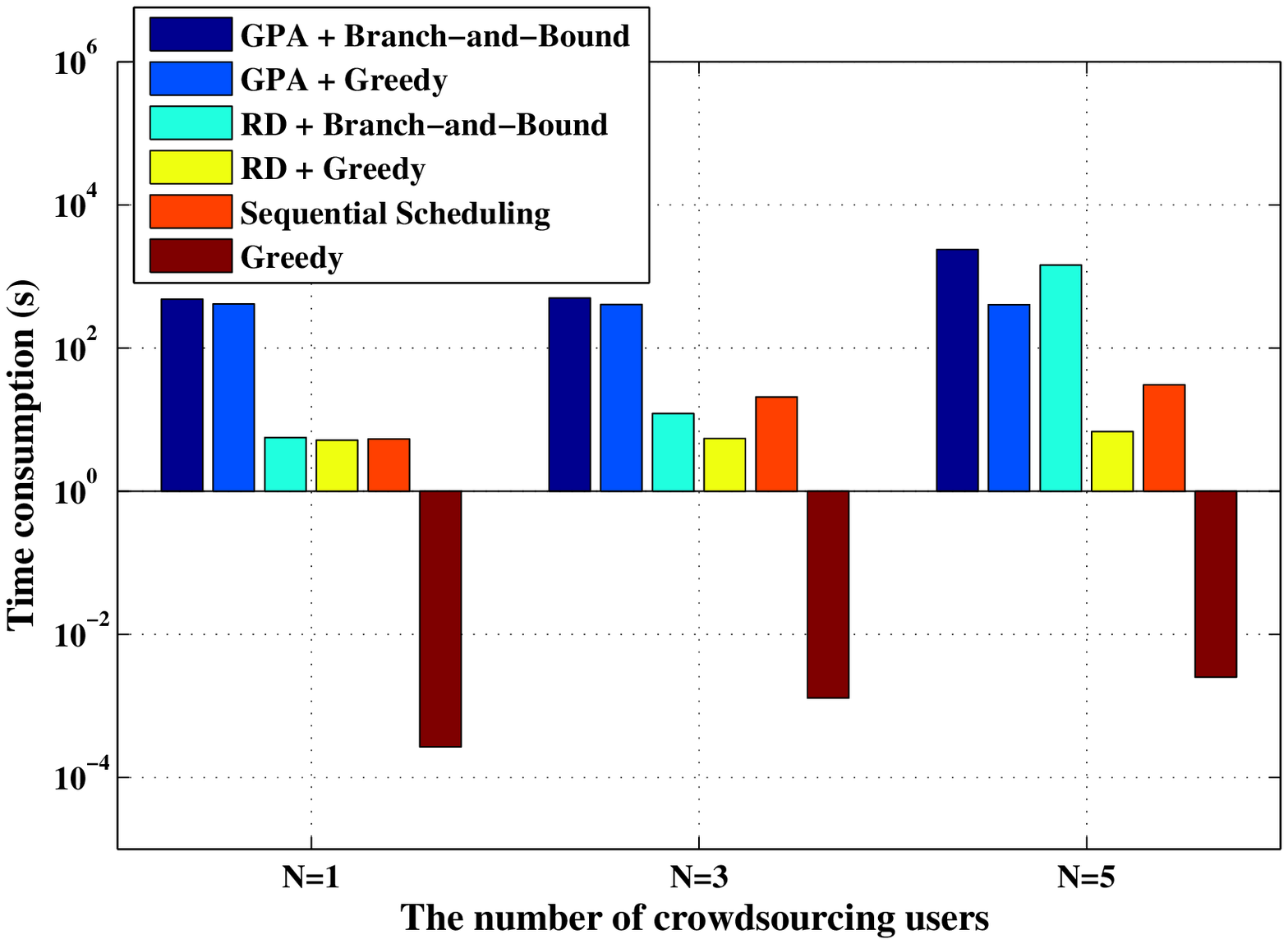}}
	\hspace{1pt}
	\subfigure[Cost allocation fairness with the increase of $J$]{
		\label{fig:mad_n} 
		\includegraphics[width=5.7cm,height=0.11\textheight]{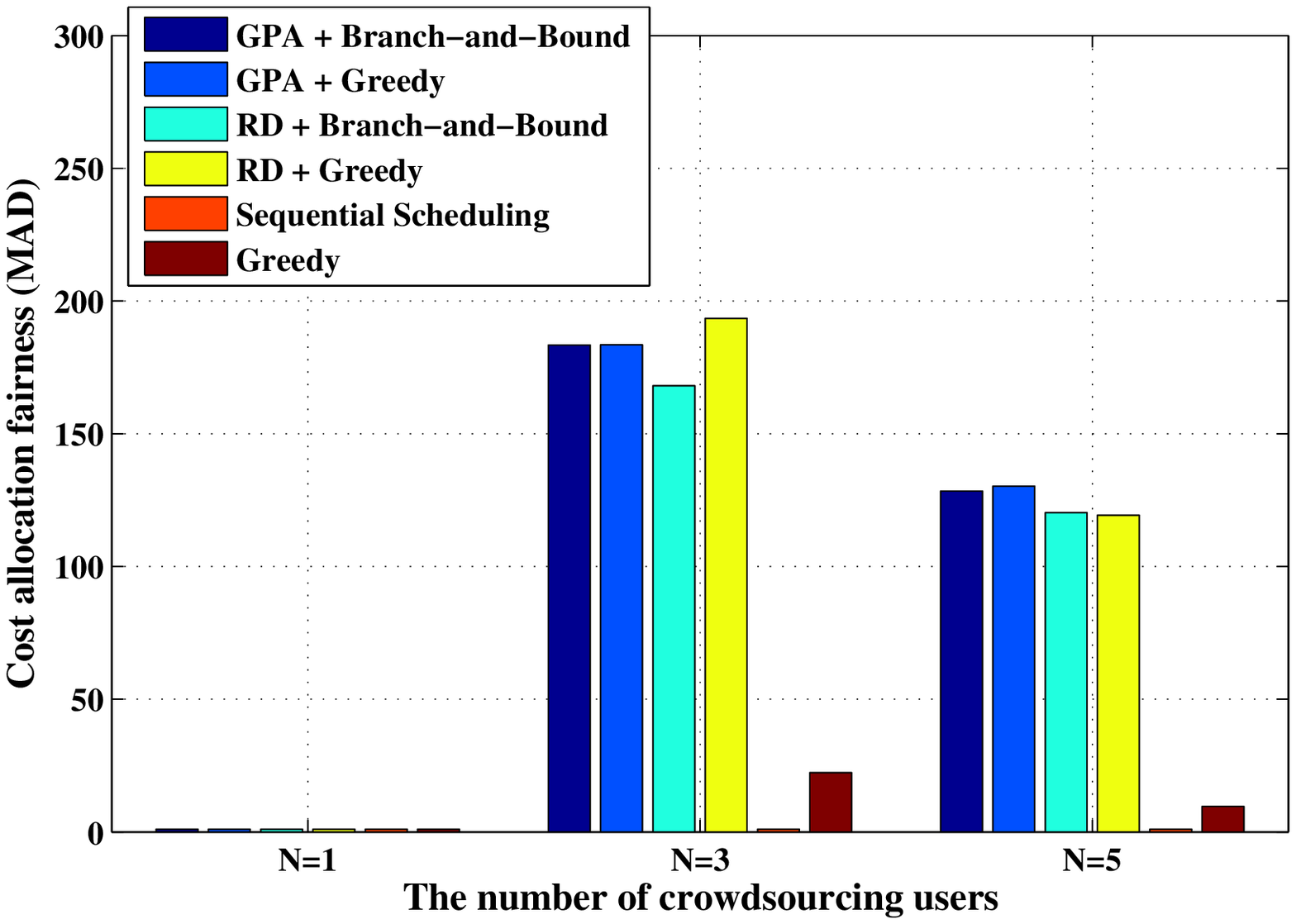}}
	\caption{The effect of users' number $J$ on the optimality, execution time and cost allocation with $\lambda_c=0.1, C=120$.}
	\label{fig:effect}
\end{figure*}

\subsection{Movement Control Influence Evaluation}
Here, we consider the effects of different user scales and different request arrival rates under the two mobility models. 
The scale of users in this area would reflect the chance one request can be served. This reflection is observed visually from Fig. \ref{fig:wolf2}, the system utility increases as the growth of user scale in this area under both two mobility models. Moreover, different from smooth growth curves under random waypoint model (Fig. \ref{fig:rwm}), the curves under campus waypoint model are rough. This may be because there exist some ``hot'' regions under campus and the users would have little possibility of meeting the requests from ``cold'' regions. When the requests in ``hot" regions are fully satisfied, the system utility would increase at a slower speed, which is just shown in Fig. \ref{fig:cwm}. Furthermore, we would find that the system utility is low on the whole and there are many requests which are not satisfied. As the total length of requests is 3718 ($\lambda_c=0.1$), the satisfaction ratio of requests is only 8.18 percent when 200 users are thrown into this area even under random waypoint model.

We also evaluate the influence of movement control. Here, we only consider single employed user case and the results are shown in Fig. \ref{fig:cost_change}. On the whole, the system utility increases with the increase of available cost. This relation can be approximately treated as an increasing concave function. This characteristic would have a valuable meaning: The system utility would be improved greatly even with less available cost. Moreover, there exist threshold values in these curves where the user is fully used. When available cost exceeds the threshold value, the system utility would not increase. Therefore, in practice, we should rationally regulate the available cost. In addition, we observe that the satisfaction ratio would be almost 50 percent with sufficient available cost even with only one employed user. Almost 4 percent of the cost at maximum utility point would achieve nearly 10 percent satisfaction ratio. These results indicate the advantage of movement control.

\subsection{Comparison with Benchmark Algorithms}



\subsubsection{Effect of the number of crowdsourcing users $J$}
In Fig. \ref{fig:utility_n}, with the increase of crowdsourcing users' number, all these algorithms can nearly achieve a better system utility. This conforms to our intuition: The increase of crowdsourcing users means that one crowdsourcing user can actually serve requests within a smaller range (less transfer cost). Therefore, the utility-cost ratio can be improved. Namely, the system utility would increase when the total available cost is fixed. Furthermore, the system utility increment decreases with the increase of crowdsourcing users. This may be because the requests seem to be more fully utilized with the increase of crowdsourcing users, which means that in reality we don't need to employ too many users in a limited area.

Fig. \ref{fig:utility_n} also indicates the advantage of clustering on the level of requests. Although clustering on the level of requests needs a little more time consumption (shown in Fig. \ref{fig:time_n}), it can better balance the requests' overlapping degree and transfer cost while the region division only focuses on the transfer cost and neglects the requests' overlapping degree. We also find that although the greedy cost allocation method can not ensure the optimality of cost allocation, the gap is relatively small compared with other algorithms due to the effect of approximately increasing concave characteristic in Fig. \ref{fig:cost_change}. Considering the low time complexity, the greedy cost allocation method seems to be more suitable for the large available cost case. Fig. \ref{fig:time_n} further illustrates the high time complexity of request clustering and branch-and-bound more visually.

We show another evaluation metric, i.e., cost allocation fairness. Cost allocation fairness is measured by mean absolute deviation (MAD) value of different employed users' allocated cost. A smaller MAD value means a better fairness. This metric can sometimes reflect the load on different users such as transfer distance or even profit. Just as shown in Fig. 9(c), all these methods have a moderate cost allocation fairness. Actually, GPA algorithm adopts normalized cut spectral clustering to partition the requests and the normalized cut spectral clustering can achieve a relatively balanced partition.

\subsubsection{Effect of total available cost $C$}
In Fig. \ref{fig:utility_cost}, all of methods achieve better system utility with the increase of total available cost, as bigger $C$ stands for more transfer opportunities. Here, Greedy has a worst performance, mainly because it does not consider the cost-effectiveness. Sequential Scheduling can have relatively better system utility due to the approximately increasing concave relation between optimal utility and allocated cost. The curves also shows the advantage of partition and cost allocation in GPA with branch-and-bound/greedy cost allocation. In addition, the approximate increasing concave relation between system utility and cost still holds in the case of multiple crowdsourcing users.

\subsection{Further Discussion}

\begin{table}[t]
	\centering
	\begin{threeparttable}[b]
		\caption{Comparison with Existing Algorithms.}
		\renewcommand\arraystretch{1.15}
		\begin{tabular}{p{2.1cm}<{\centering}p{1.8cm}<{\centering}p{1.1cm}<{\centering}p{2.1cm}<{\centering}}
			\toprule
			    & System Utility  &  Time  & Fairness (MAD) \\
			\midrule
			SOF \cite{li2014message}          & 101          & \textbf{129} & \textbf{98} \\
			GBB \cite{li2014message}          & 112          & 622  & 155 \\
			User-Centric \cite{chen2017cache} & 506          & 1850 & 127 \\
			MicroCast \cite{le2015microcast}  & 542          & 3410 & 130 \\
			\textbf{GPA+B\&B}                 & \textbf{628} & 3910 & 136 \\
			\bottomrule
		\end{tabular}
		\label{table:compared_algs}
	\end{threeparttable}
\end{table}

Based on the comparison with benchmark algorithms, it is intuitive to find that GPA with Branch-and-Bound (abbrv. GPA+B\&B) has the optimal system utility. Here, we compare GPA+B\&B with the existing approaches, including SOF \cite{li2014message}, GBB \cite{li2014message}, User-Centric \cite{chen2017cache}, and MicroCast \cite{le2015microcast}. In order to evaluate the scheduling strategies in the same condition, we modify these algorithms and set the key parameters (e.g., $ N=5 $).  Table~\ref{table:compared_algs} shows the concrete results of five algorithms from three aspects. It indicates that our algorithm, i.e., GPA+B\&B, has the best system utility and needs the more time due to the relative complex operations. Although SOF has the least time and the best fairness, its system utility is worst due to selecting an available neighbor at random. Therefore, in the future, we plan to improve our algorithms to reduce the computational complexity.

\section{Conclusion}
\label{sec:conclude}

Mobile devices speed up the development of new video applications and put forward higher demands on current mobile access network. In this paper, we adopt the crowdsourcing paradigm to offer incentive for guiding the movement of recruited crowdsourcing users and facilitate the optimization of the movement decision for high-quality video enhancement. We analyze the influence of crowdsourcing style on the improvement of D2D communication. We formulate the movement control decision as a cost-constrained user recruitment optimization problem. We study in detail how to optimize movement control decision for single and multiple users. The effectiveness of our algorithms is evaluated by simulations. The results demonstrate the crowdsourcing style can guarantee a higher D2D communication efficiency for video enhancement.

\section*{Acknowledges}
Our thanks to the reviewers for their constructive comments and suggestions to improve the quality of the manuscript. This work was supported in part by the National Key Research and Development Project under Grant No. 2020YFB1707601, the Natural Science Research of Jiangsu Higher Education Institutions of China under Grant No. 21KJB520036, Jiangsu Provincial Program for Innovation \& Entrepreneurship under Grant No. JSSCBS20210330, the Open Topic of State Key Lab. for Novel Software Technology, Nanjing University, P.R. China under Grant No. KFKT2021B35. Dr Xuyun Zhang is the recipient of an ARC DECRA (project No. DE210101458) funded by the Australian Government. 

\bibliographystyle{ACM-Reference-Format}
\balance
\bibliography{sample-base}

\end{document}